\documentclass[a4paper, 11pt]{article}
\usepackage[T1]{fontenc}
\usepackage[utf8]{inputenc}

\usepackage[margin=2.5cm]{geometry}
\usepackage{amsmath}
\usepackage{amssymb}
\usepackage{amsthm}
\usepackage{amsfonts}
\usepackage[hidelinks]{hyperref}

\usepackage{bm}
\usepackage{csquotes}
\usepackage{enumitem}
\usepackage{mathtools}
\usepackage[normalem]{ulem}
\usepackage{color, colortbl}
\usepackage[ruled, vlined, linesnumbered]{algorithm2e}
\usepackage{multirow}
\usepackage{cite}
\usepackage{pgfplots}
\usepackage{tikz}
\usepackage{etoolbox}

\usetikzlibrary{matrix,decorations.pathreplacing, calc, positioning, fit, patterns}
\SetKwInput{KwData}{Input}
\SetKwInput{KwResult}{Output}
\setlist[itemize]{parsep=1pt, topsep=1pt}
\setenumerate{topsep=0pt,itemsep=-1ex,partopsep=1ex,parsep=1ex}
\hypersetup{colorlinks=false, linkcolor=red}
\theoremstyle{plain}
\newtheorem{theorem}{Theorem}[section]
\newtheorem*{theorem*}{Theorem}
\newtheorem{proposition}[theorem]{Proposition}
\newtheorem*{proposition*}{Proposition}
\newtheorem{corollary}[theorem]{Corollary}
\newtheorem*{corollary*}{Corollary}

\newtheorem*{lemma*}{Lemma}

\theoremstyle{definition}

\newtheorem*{example*}{Example}

\newcommand\FF{\mathbb{F}}

\newcommand\calA{\mathcal{A}}

\newcommand\calC{\mathcal{C}}

\newcommand\calI{\mathcal{I}}
\newcommand\calJ{\mathcal{J}}

\newcommand\calO{\mathcal{O}}

\newcommand\calU{\mathcal{U}}

\newcommand\bfa{{\bm{a}}}

\newcommand\bfc{{\bm{c}}}

\newcommand\bfx{{\bm{x}}}
\newcommand\bfy{{\bm{y}}}

\newcommand\bfzero{{\bm{0}}}

\newcommand\bfA{{\bm{A}}}

\newcommand\bfD{{\bm{D}}}
\newcommand\bfE{{\bm{E}}}

\newcommand\bfM{{\bm{M}}}

\newcommand\bfZ{{\bm{Z}}}

\newcommand\bfQ{{\bm{Q}}}

\newcommand\bfX{{\bm{X}}}
\newcommand\bfY{{\bm{Y}}}

\DeclareMathOperator{\rk}{rk}

\renewcommand\epsilon{\varepsilon}
\newcommand\mydef{\coloneqq}
\newcommand\qbin[3]{
  \left[\begin{smallmatrix} #1 \\[0.3em] #2 \end{smallmatrix}\right]_{#3}
}
\newcommand\bigqbin[3]{
  \left[\begin{matrix} \,#1\, \\ \,#2\, \end{matrix}\right]_{#3}
}

\makeatletter
\patchcmd{\maketitle}{\@fnsymbol}{\@alph}{}{}  % Footnote numbers from symbols to small letters
\makeatother

\title{On the privacy of a code-based single-server\\
  computational PIR scheme}

\author{
    Sarah Bordage\thanks{LIX, CNRS UMR 7161, Ecole Polytechnique, Institut Polytechnique de Paris \& Inria, 91120 Palaiseau, France. \texttt{sarah.bordage@lix.polytechnique.fr}}
    \and
    Julien Lavauzelle\thanks{Univ. Rennes, CNRS, IRMAR -- UMR 6625, F-35000 Rennes, France. \texttt{julien.lavauzelle@univ-rennes1.fr}}
}

\date{}

\date{\today}
\begin{document}

\maketitle

\begin{abstract}
  We show that the single-server computational PIR protocol proposed by Holzbaur, Hollanti and Wachter-Zeh in~\cite{HolzbaurHW20} is not private, in the sense that the server can recover in polynomial time the index of the desired file with very high probability. The attack relies on the following observation. Removing rows of the query matrix corresponding to the desired file yields a large decrease of the dimension over $\mathbb{F}_q$ of the vector space spanned by the rows of this punctured matrix. Such a dimension loss only shows up with negligible probability when rows unrelated to the requested file are deleted.
\end{abstract}

\section{Introduction}

Private information retrieval (PIR) enables a user to retrieve an entry of a database without revealing to the storage system the identity of the requested entry. Two security models have been introduced for PIR schemes. First, the seminal work of Chor \emph{et al.}~\cite{ChorGKS95} proposes information-theoretical security, in the sense that absolutely no information leaks about the identity of the desired item. A trivial solution, commonly referred as the trivial PIR scheme, is to require the storage system to send the whole database to the user. As a matter of fact, the authors of~\cite{ChorGKS95} also proved that, in the single-server information theoretic setting, one cannot expect to achieve communication complexity better than the trivial solution. The second security model circumvents this limit and allows the more practical use of a single server by relaxing the privacy requirement. In this model, the storage system is assumed to be computationally bounded: informally, recovering the identity of the desired item must require an attacker to invest unreachable computational effort. So-called computationally private information retrieval (cPIR) was firstly introduced in~\cite{ChorG97, KushilevitzO97}, and subsequent constructions~\cite{CachinMS99, YiKPB13, GentryR05, KiayiasLLPT15, LipmaaP17}  were then proposed. Aguilar \emph{et al.\ }proved the potential practicality of cPIR~\cite{AguilarBFK16}, but the question of building efficient cPIR protocols remains widely open. Indeed, the computational complexity of existing cPIR schemes is the most important barriers to implementation.

In this paper we focus on the recent single-server cPIR protocol proposed by Holzbaur, Hollanti and Wachter-Zeh in~\cite{HolzbaurHW20}, which relies on computational assumptions in coding theory. We prove that this scheme is not private: we present an algorithm which recovers the identity of the file in polynomial time and with very high probability, when given as input the query produced by the user. We implemented our attack, which runs in a few minutes on a standard laptop. The attack requires the number of files stored in the database to be large enough, namely lower-bounded by some function of the scheme parameters. It turns out that this condition is fulfilled for meaningful parameters of the scheme. Indeed, we show that if this lower bound is not satisfied, then the communication complexity of the cPIR scheme gets very close to the one of the trivial PIR protocol.

The paper is organized as follows. In Section~\ref{sec:description} we describe the scheme proposed in~\cite{HolzbaurHW20}. The attack is presented and proved in Section~\ref{sec:attack} and followed by a short discussion.

\section{Description of the cPIR scheme proposed in \cite{HolzbaurHW20}}
\label{sec:description}

In this section, we briefly describe the PIR scheme proposed in~\cite{HolzbaurHW20}.

\subsection{Notation and definitions}

Let us denote $[a, b] \mydef \{ a, a+1, \dots, b\}$ and $\FF_q$ the finite field with $q$ elements. The extension field $\FF_{q^s}$ is also a vector space of dimension $s$ over $\FF_q$. If $\Gamma = \{ \gamma_1, \dots, \gamma_v \} \subset \FF_{q^s}$ is a family of linearly independent vectors over $\FF_q$, then we denote $\langle \gamma_1, \dots, \gamma_v \rangle_{\FF_q} \subseteq \FF_{q^s}$ the vector space of dimension $v$ over $\FF_q$ which is generated by the elements in $\Gamma$. We also define $\psi_\Gamma : \FF_{q^s} \to \langle \gamma_1, \dots, \gamma_v \rangle_{\FF_q}$ the corresponding projection map.

For a vector $\bfx = (x_1,\dots, x_t) \in \FF_{q^s}^t$ and an ordered subset $\calJ \subset [1,n]$ of size $t$,  we denote $\phi_\calJ(\bfx) \in \FF_q^n$ the extension of the vector $\bfx$ with zeroes at indices $j \notin \calJ$. For instance, if $n = 5$ and $\calJ = \{1,4\}$, then $\phi_{\{1,4\}}((x_1,x_2)) = (x_1, 0, 0, x_2, 0)$. This map is extended to matrices by applying $\phi_\calJ$ row-wise. Conversely, if $\bfx = (x_1,\dots, x_n) \in \FF_{q^s}^n$, the punctured vector $\bfx_\calJ$ is  $\bfx_\calJ \mydef (x_{j_1}, \dots, x_{j_t}) \in \FF_{q^s}^t$. For a subset $\calA \subset \FF_{q^s}^n$, one writes $\calA_\calJ \mydef \{ \bfa_\calJ \mid \bfa \in \calA \}$.

Given a linear code $\calC \subseteq \FF_{q^s}^n$ of dimension $k$, an information set  for $\calC$ is a subset $\calI \subset [1,n]$ of size $k$ such that $\calC_{\calI} = \FF_{q^s}^k$. Finally, given a matrix $\bfM \in \FF_{q^s}^{r \times n}$, we define the rank over $\FF_q$ of $\bfM$, denoted $\rk_{\FF_q}(\bfM)$, as the dimension over $\FF_q$ of the vector space generated by the rows of $\bfM$. Notice that $\rk_{\FF_q}(\bfM) \le \min \{ns, r\}$.

\subsection{System model}

In~\cite{HolzbaurHW20}, it is assumed that a single server stores $m$ large files $\bfX^1, \dots, \bfX^m$ of the same size. In particular, for each $i \in [1,m]$ the symbols of the $i$-th file are arranged in a matrix $\bfX^i \in \FF_q^{L \times (s-v)(n-k)}$, for some $L \ge 1$. For convenience, we denote $\delta \mydef (s-v)(n-k)$. Notice that integers $m$, $s$, $v$, $n$, $k$, $q$, $L$ are known to both the user and the server.

\subsection{Queries}

We here assume that the user wants to retrieve a specific file $\bfX^i$, for a given $i \in [1,m]$. In order to generate a corresponding query $\bfQ^i$, the user samples uniformly at random:
\begin{itemize}[label=--]
  \item a code $\calC \subseteq \FF_{q^s}^n$ of dimension $k$,
  \item a information set $\calI \subset [1,n]$ for $\calC$,
  \item a basis $\{\gamma_1, \dots, \gamma_s\}$ of $\FF_{q^s}$ over $\FF_q$, and sets $V \mydef \langle \gamma_1, \dots, \gamma_v \rangle_{\FF_q}$ and $W \mydef \langle \gamma_{v+1}, \dots, \gamma_s \rangle_{\FF_q}$,
  \item a matrix $\bfD \in \FF_{q^s}^{m \delta \times n}$ such that each row of $\bfD$ is a codeword in $\calC$,
  \item a matrix $\bfE \in V^{m \delta \times n}$ such that the $j$-th column of  $\bfE$ is zero if $j \in \calI$, and lies in $V^{m \delta}$ otherwise,
  \item a matrix $\bfZ^i \in W^{m \delta \times n}$ such that the submatrix
    \[
    \bfZ^i_{[ i\delta +1, (i+1)\delta] \times \overline{\calI}} \in W^{\delta \times (n-k)}
    \]
    has rank $\delta$ over $\FF_q$, and such that all remaining entries of $\bfZ$ are zeroes.
\end{itemize}

Eventually, the user sends $\bfQ^i \mydef \bfD + \bfE + \bfZ^i  \in \FF_{q^s}^{m \delta \times n}$ to the server as a query. See Figure~\ref{fig:query} for an illustration.

\begin{figure}[t]
  \centering
  \begin{tikzpicture}[scale=0.85]
    \draw[fill=black!8] (0,0) rectangle (3,8) node[midway] {$\bfD$};
    \draw (4,0) rectangle (7,8) node[midway] {$\bfE$};
    \draw (8,0) rectangle (11,8) node[midway] {$\bfZ$};

    \draw[dashed] (0,6.4) rectangle (3,6.8) node[midway] {\scriptsize $\bfc \in \calC$};
    
    \draw[pattern=north west lines, pattern color=blue] (4,0) rectangle (4.9,8);
    \draw[pattern=north west lines, pattern color=blue] (5.8,0) rectangle (6.1,8);
    \draw[pattern=north west lines, pattern color=blue] (6.7,0) rectangle (7,8);

    \draw[pattern=crosshatch dots, pattern color=red] (8,1.5) rectangle (8.9,3);
    \draw[pattern=crosshatch dots, pattern color=red] (9.8,1.5) rectangle (10.1,3);
    \draw[pattern=crosshatch dots, pattern color=red] (10.7,1.5) rectangle (11,3);

    \node[left] at (0,4) {$\bfQ^i \;=$};
    \node at (3.5,4) {$+$};
    \node at (7.5,4) {$+$};

    \draw [decorate, decoration={brace, raise=3pt}] (0,8) -- (3,8) node [midway, above=5pt] {\footnotesize $n$};
    
    \draw [decorate, decoration={brace, mirror, raise=3pt}] (11,0) -- (11,8) node [midway, rotate=270, yshift=12pt] {\footnotesize $m\delta = m(n-k)(s-v)$};

    \draw [decorate, decoration={brace, raise=3pt}] (8,1.5) -- (8,3) node [midway, rotate=90, yshift=12pt] {\scriptsize $[i\delta+1, (i+1)\delta]$};

  \end{tikzpicture}
  \caption{\label{fig:query}Illustration of query matrix $\bfQ^i$ for a random decomposition $\FF_{q^s} = V \oplus W$. The region filled uniformly in gray represents elements in $\FF_{q^s}$; the blue hashed region refers to elements in $V$; the red dotted region contains elements in $W$.}

\end{figure}
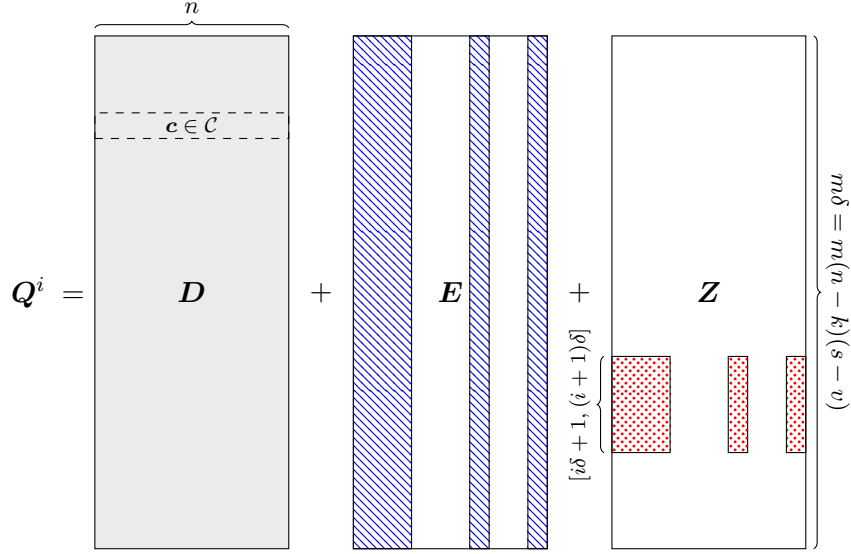

\subsection{Response}

The server computes and sends back the result of the matrix product $\bfA^i = [ \bfX^1, \dots, \bfX^m ] \cdot \bfQ^i \in \FF_{q^s}^{L \times n}$ to the user.

\subsection{Decoding}

Let us decompose the matrix $\bfQ^i \in \FF_{q^s}^{m \delta \times n}$ as a stack of $m$ submatrices $\bfQ^i_1, \dots, \bfQ^i_m \in \FF_{q^s}^{\delta \times n}$. One can proceed similarly for $\bfD$, $\bfE$ and $\bfZ^i$. Then we have:
\[
\bfA^i = \sum_{r=1}^m \bfX^r \cdot \bfQ^i_r = \sum_{r=1}^m \bfX^r \cdot \bfD_r + \sum_{r=1}^m \bfX^r \cdot (\bfE_r +  \bfZ^i_r)\,.
\]
The rows of matrix $\sum_{r=1}^m \bfX^r \cdot \bfD_r$ all lie in $\calC$. By inverting a linear system on the information set $\calI$, the user can thus recover
\[
\bfY = \bfA^i - \sum_{r=1}^m \bfX^r \cdot \bfD_r = \sum_{r=1}^m \bfX^r \cdot (\bfE_r +  \bfZ^i_r) = \left( \sum_{r=1}^m \bfX^r \cdot \bfE_r \right) + \bfX^i \cdot \Delta
\]
where $\Delta \mydef \bfZ^i_{[i\delta+1, (i+1)\delta] \times [1,n]}$.

It remains to notice that, for a given basis $\mathcal{W}$ of $W$, we have $\psi_{\mathcal{W}}(\bfY) = \bfX^i \cdot \Delta$. Since $\rk_{\FF_q}(\Delta) = \delta$, the user can eventually retrieve $\bfX^i$ from $\bfX^i \cdot \Delta$.

\section{An efficient attack based on the $\FF_q$-rank of submatrices}
\label{sec:attack}

\subsection{Presentation of the attack}
\label{subsec:attack}

Informally, the attack relies on the following observation: for a large enough number of files and with high probability, the $\FF_q$-rank of $\bfD + \bfE$ is much lower than the rank of $\bfQ^i$. Hence, if we denote $\bfQ^i[j]$ the submatrix of $\bfQ^i$ obtained after deletion of rows $[j\delta+1, (j+1)\delta]$, then  one can easily distinguish between the two following cases:
\begin{enumerate}
  \item $\bfQ^i[i]$ (in which case the only non-zero component $\Delta$ of $\bfZ^i$ has been removed), and 
  \item $\bfQ^i[j]$ for $j \in [1,m] \setminus \{i\}$ (in which case the component $\Delta$ still remains).
\end{enumerate}

Let us first prove a first result concerning the structure of the matrix $\bfQ^i$ over $\FF_q$.

\begin{proposition}
  \label{prop:decomposition}
  Let $\FF_{q^s} = V \oplus W$, $\calC \subseteq \FF_{q^s}^n$ and $\calI$ be chosen as in Section~\ref{sec:description}. Then, we have the following decomposition of $\FF_{q^s}^n$ into $\FF_q$-linear spaces:
  \[
  \calC \oplus \phi_{\overline{\calI}}(V^{n-k}) \oplus  \phi_{\overline{\calI}}(W^{n-k}) = \FF_{q^s}^n\,.
  \]
  Moreover, any query $\bfQ^i = \bfD + \bfE + \bfZ^i$ satisfies:
  \[
  \langle \bfD \rangle_{\FF_q} \subseteq \calC, \quad\quad \langle \bfE \rangle_{\FF_q} \subseteq \phi_{\overline{\calI}}(V^{n-k})\,,  \quad\quad \text{ and }  \quad\quad \langle \bfZ^i \rangle_{\FF_q} \subseteq \phi_{\overline{\calI}}(W^{n-k})\,.
  \]
\end{proposition}

\begin{proof}
  The set $\calI \subset [1,n]$ is an information set for $\calC \subseteq \FF_{q^s}^n$, hence it holds $\calC \oplus \phi_{\overline{\calI}}(\FF_q^{n-k}) = \FF_{q^s}^n$ as $\FF_{q^s}$-linear spaces. This equality holds \emph{a fortiori} as $\FF_q$-linear spaces. We also have $V \oplus W = \FF_{q^s}$, and since $\phi_{\overline{\calI}}$ is $\FF_q$-linear, it follows that $\calC \oplus \phi_{\overline{\calI}}(V^{n-k}) \oplus  \phi_{\overline{\calI}}(W^{n-k}) = \FF_{q^s}^n$.
\end{proof}

One can now notice that $\bfZ^i[i] = \bfzero$, hence $\bfQ^i[i] = \bfD[i] + \bfE[i]$. As a corollary, observe that the rank of $\bfQ^i[i]$ is remarkably low.

\begin{corollary}
  \label{cor:low-rank}Let us denote $k_0 \mydef ks + v(n-k) = sn - \delta$. 
  For every $i \in [1,m]$, we have $\rk_{\FF_q}(\bfQ^i[i]) \le k_0$.
\end{corollary}
\begin{proof}
  This is a direct consequence of the fact that $\dim_{\FF_q}(\calC) = ks$ and  $\dim_{\FF_q}(V^{n-k}) = v(n-k)$.
\end{proof}

Let us now characterize the rank of $\bfQ^i[j]$ for $j \in [1,m] \setminus \{ i\}$. Due to Proposition~\ref{prop:decomposition},  we have
\[
\rk_{\FF_q}(\bfQ^i[j]) = \rk_{\FF_q}(\bfD[j] + \bfE[j]) + \rk_{\FF_q}(\bfZ^i[j]) = \rk_{\FF_q}(\bfD[j] + \bfE[j]) + \delta
\]
since matrix $\Delta$ has rank $\delta$ over $\FF_q$, by construction.

Hence, it remains to compute the probability that $\rk_{\FF_q}(\bfD[j] + \bfE[j])$ does not shrink too much to enable an attacker to distinguish between $\rk_{\FF_q}(\bfQ^i[j])$ and $\rk_{\FF_q}(\bfQ^i[i])$.

For $a \le b$, let us denote
\[
\bigqbin{b}{a}{q} \mydef \frac{(q^b-1)(q^b-q) \cdots (q^b-q^{a-1})}{(q^a-1)(q^a-q) \cdots (q^a-q^{a-1})}
\]
the Gaussian, or $q$-binomial, coefficient which counts the number of $\FF_q$-linear spaces of dimension $a$ contained in a fixed $b$-dimensional linear space over $\FF_q$. 

\begin{proposition}
  \label{prop:proba-low-rank}
  Let $\bfQ^i = \bfD + \bfE + \bfZ^i$ be a query generated as in Section~\ref{sec:description}. Let also $j \ne i$ and $k_0 = sn - \delta$. Then we have:
  \[
  \Pr\Big(\rk_{\FF_q}(\bfD[j] + \bfE[j]) \le k_0 - \delta \Big) \le \bigqbin{k_0}{k_0-\delta}{q} \cdot q^{-\delta^2(m-1)}\,,
    \]
    where the probability is taken over the randomness of the query generation.
\end{proposition}

\begin{proof}
  Let us denote $\calU \mydef \calC \oplus \phi_{\overline{I}}(V^{n-k})$ and recall that $\calU$ is a $\FF_q$-linear space of dimension $k_0$. During the generation of the query $\bfQ^i$, each row of $\bfD + \bfE$ is actually a vector from $\calU$ picked uniformly at random. Hence, the probability we aim at bounding is exactly
  \[
  p \mydef \Pr\Big( \exists \calA \subset \calU, \dim_{\FF_q}(\calA) = k_0 - \delta \mid \forall \bfy \in \textsf{Rows}(\bfD[j] + \bfE[j]), \bfy \in \calA \Big)\,,
  \]
  where $\textsf{Rows}(\bfD[j] + \bfE[j])$ represents the set of rows of $\bfD[j] + \bfE[j]$, seen as vectors of length $ns$ over $\FF_q$. Let us denote ${\rm Gr}_{\calU}(k_0 - \delta)$ the set of subspaces of dimension $k_0 - \delta$ included in $\calU$. By union bound we get
  \[
  p \le \sum_{\calA \in {\rm Gr}_{\calU}(k_0 - \delta)}  \Pr\Big( \forall \bfy \in \textsf{Rows}(\bfD[j] + \bfE[j]), \bfy \in \calA \Big)\,.
  \]
  Rows $\bfy \in \textsf{Rows}(\bfD[j] + \bfE[j])$ are vectors picked uniformly and independently in $\calU$. Thus, this yields
  \[
  p \le \sum_{\calA \in {\rm Gr}_{\calU}(k_0 - \delta)} \left( \prod_{t=1}^{(m-1)\delta} \Pr( \bfy \in \calA \mid \bfy \leftarrow \calU) \right)
  \]
  Note that ${\rm Gr}_{\calU}(k_0 - \delta)$ has cardinality $\qbin{k_0}{k_0-\delta}{q}$. Hence, 
  \[
  p \le \bigqbin{k_0}{k_0-\delta}{q} \cdot \prod_{t=1}^{(m-1)\delta} q^{-\delta} = \bigqbin{k_0}{k_0-\delta}{q} \cdot q^{-\delta^2(m-1)}\,.
  \]
\end{proof}

Notice that a rough upper bound for the Gaussian coefficient $\qbin{k_0}{k_0-\delta}{q}$ is $q^{(\delta+1)(k_0-\delta)}$. Hence, the upper bound given in Proposition~\ref{prop:proba-low-rank} is meaningful as soon as $(\delta+1)(k_0-\delta) \le \delta^2(m-1)$. Thus, let us define
\[
m_0 \mydef 1 + \left\lceil \frac{(\delta+1)(k_0-\delta)}{\delta^2} \right\rceil = 1 + \left\lceil \Big(1 + \frac{1}{\delta}\Big)\Big(\frac{sn}{\delta}-2\Big) \right\rceil\,.
\]

We can now state the main result of the paper. The polynomial time algorithm we propose as an attack to the cPIR scheme is given as a proof of our main theorem.

\begin{theorem}
  \label{thm:main-thm}
  Let $\bfQ^i = \bfD + \bfE + \bfZ^i \in \FF_{q^s}^{m \delta \times n}$ be a query generated as in Section~\ref{sec:description}, and assume that $m \ge m_0 = 1 + \lceil \frac{(\delta+1)(k_0-\delta)}{\delta^2}\rceil$. There exists an algorithm running in $\calO(m^2(sn)^3)$ operations over $\FF_q$, which recovers the index $i$ when given as input $\bfQ^i$ with probability at least
  \[
  1 - q^{-(m-m_0)\delta^2},
  \]
  where the probability is taken over the randomness of the query generation.
\end{theorem}

\begin{proof}
  The algorithm consists in the following. Given the query $\bfQ^i$, first compute the $\FF_q$-rank of submatrices $\bfQ^i[j] \in \FF_{q^s}^{(m-1)\delta \times n}$ for every $j \in [1,m]$. Then, output the index $j^* \in [1,m]$ (if unique) such that $\rk_{\FF_q}(\bfQ^i[j^*]) \le k_0$.  Notice that the $\FF_q$-rank of matrices can be computed with any basis of $\FF_{q^m}/\FF_q$, and thus, independently of the knowledge of the basis $\{\gamma_1, \dots, \gamma_s\}$ chosen by the user. 

  From Corollary~\ref{cor:low-rank}, we indeed have $\rk_{\FF_q}(\bfQ^i[i]) \le k_0$. Moreover, from Proposition~\ref{prop:proba-low-rank} and the discussion above, the probability that $\rk_{\FF_q}(\bfQ^i[j]) \le k_0$ for some $j \ne i$ is upper bounded by
  \[
  \bigqbin{k_0}{k_0-\delta}{q} \cdot q^{-\delta^2(m-1)} \le q^{-(m-m_0)\delta^2 + (\delta+1)(sn-2\delta)- (m_0-1)\delta^2} \le  q^{-(m-m_0)\delta^2},
  \]
  by definition of $m_0$.

  The running time of the algorithm is in $\calO(m^2(sn)^3)$ since it consists in computing $m$ times the $\FF_q$-rank of a matrix of size $(m-1)\delta \times n$ over $\FF_{q^s}$, where $\delta \le sn$.
\end{proof}

\subsection{Discussion}

%\paragraph{On the restriction $m \ge m_0$.} The algorithm presented in Theorem~\ref{thm:main-thm} breaks the CPIR system from~\cite{HolzbaurHW20} if the number $m$ of files stored by the server satisfies $m \ge m_0$ (in particular, we prove that this cPIR system cannot support an unbounded number of files). Let us now show that, for $m < m_0$, this cPIR system admits a very low PIR rate, close to the trivial rate. Recall that the PIR rate of a PIR system is the ratio between the bisize of the desired file and the whole number of bits downloaded in the retrieval protocol. We say the PIR rate is trivial if it corresponds to downloading the entire database, \emph{i.e.} $R_{\rm PIR} = 1/m$.
%
%In~\cite{HolzbaurHW20} the authors prove that for large files,
%\[
%R_{\rm PIR} \simeq \frac{\delta}{sn}\,.
%\] \sarah{$R_{\rm PIR}$ is already used in the text to refer to another rate.}
%
%Hence, if $m < m_0$, we get:
%\[
%m - 1 \; \lesssim \; \left(1 + \frac{1}{\delta}\right)\left(\frac{1}{R_{\rm PIR}} - 2\right)\,,
%\]
%hence,
%\[
%R_{\rm PIR} \; \lesssim \; \frac{1 + \frac{1}{\delta}}{m + 1 + \frac{2}{\delta}} \,.
%\]

In this paragraph, we discuss the necessary condition $m \ge m_0$ for the attack to work. We show that this condition is fulfilled for any relevant parameter of the system. More specifically, Theorem~\ref{thm:main-thm} proposes an efficient attack against the cPIR system from~\cite{HolzbaurHW20} if the number $m$ of files stored by the server is at least $m_0 = 1 + \lceil (1 + \frac{1}{\delta})(\frac{sn}{\delta} - 2)\rceil$. In particular, this cPIR system cannot support an unbounded number of files.

Let us discuss the efficiency of the PIR scheme in the converse case. Recall that the PIR rate of a PIR system is the ratio between the bit size of the desired file and the total communication complexity. The PIR rate corresponding to the trivial PIR protocol is equal to $1/m$. A common assumption for cPIR schemes is to treat the query size as negligible compared to the size of the stored files. For the considered cPIR scheme, this boils down to assuming that $L \gg sn$. The PIR rate may then be approximated as the ratio of the size of a file over the number of bits downloaded in the retrieval protocol. In any case, the bandwidth required for the download of the entire database gives an upper bound on the cost one is willing to afford in terms of communication. However, if $m < m_0$, it turns out that the PIR rate of that cPIR system drops close to the rate of the trivial PIR protocol, as we show next.

In~\cite{HolzbaurHW20}, the authors prove that for large files, \emph{i.e.} $L \gg \delta m$, we have 
\[
R_{\rm PIR} \simeq \frac{\delta}{sn}\,.
\]

Hence, if $m < m_0$, we get:
\[
m - 1 \; \lesssim \; \left(1 + \frac{1}{\delta}\right)\left(\frac{1}{R_{\rm PIR}} - 2\right)\,,
\]
and it follows that
\[
R_{\rm PIR} \; \lesssim \; \frac{1 + \frac{1}{\delta}}{m + 1 + \frac{2}{\delta}} \,.
\]
Thus, for $m<m_0$ the PIR rate is bounded by $\frac{2}{m+3}$. Moreover, for large values of $\delta$, one gets $R_{\rm PIR} = \calO\big(\frac{1}{m}(1 + \frac{1}{\delta})\big)$, \emph{i.e.} in this context the protocol is not significantly better than the trivial entire download of the database.

\section*{Acknowledgments}
The first author benefits from the support of the Chair \enquote{Blockchain \& B2B Platforms}, led by l'X -- {\'E}cole Polytechnique and the Fondation de l'{\'E}cole Polytechnique, sponsored by Capgemini. The second author is funded by French \emph{Direction G{\'e}n{\'e}rale l'Armement}, through the \emph{P{\^o}le d'excellence cyber}.

\bibliographystyle{alpha}
\bibliography{biblio}
\end{document}